\documentclass[11pt,a4paper]{article}

\usepackage{amsmath,amssymb,amsthm}
\usepackage{geometry}
\usepackage[colorlinks=true,citecolor=blue,linkcolor=blue,urlcolor=blue]{hyperref} 
\usepackage{mathtools}
\usepackage{enumitem}
\usepackage{graphicx}
\usepackage{authblk}

\newtheorem{theorem}{Theorem}[section]
\newtheorem{lemma}[theorem]{Lemma}
\newtheorem{proposition}[theorem]{Proposition}
\newtheorem{corollary}[theorem]{Corollary}
\newtheorem{definition}[theorem]{Definition}
\newtheorem{remark}[theorem]{Remark}
\newtheorem{example}[theorem]{Example}

\DeclareMathOperator{\tr}{tr}
\DeclareMathOperator{\ev}{ev}
\DeclareMathOperator{\Pair}{Pair}
\newcommand{\NC}{\mathrm{NC}_2}
\newcommand{\xh}{\hat{x}}

\newcommand{\cW}{\mathcal{W}}
\newcommand{\cL}{\mathcal{L}}
\newcommand{\cP}{\mathcal{P}}
\newcommand{\dd}{\mathrm{d}}
\newcommand{\shuffle}{\mathbin{\text{\footnotesize$\sqcup\!\sqcup$}}}

\title{Topological Recursion and Quantum Path Signatures}

\author[1]{In\^es Aniceto}
\author[2]{Thomas Cass}
\author[3]{Samuel Crew\thanks{Corresponding author  \texttt{samuel.c.crew@gmail.com}.}}
\affil[1]{School of Mathematical Sciences, University of Southampton, UK}
\affil[2]{Department of Mathematics, Imperial College
London, UK}
\affil[3]{QPerfect, Strasbourg, France}
\date{\today}

\begin{document}
\maketitle

\begin{abstract}
We introduce a non-commutative Laplace transform between functionals on path space and formal series in a tensor algebra, under which a natural convolution of path functionals becomes an algebraic product of series. Applying it to a random unitary matrix-valued path development---the quantum path signature---we show that the governing planar loop equations take a non-commutative spectral form. We then extend the loop equations to a $1/N$ genus expansion, organised by topological recursion, and obtain a hierarchy of integral equations on path space for the corrections.
\end{abstract}

\section{Introduction}\label{sec:intro}

The path signature \cite{chen1954iterated,lyons2007differential} is a collection of iterated integrals of a path $\gamma\colon[0,T]\to\mathbb{R}^d$. It is a central object in rough path theory \cite{lyons1998differential} with broad applications to machine learning. It serves as a universal nonlinear feature map for sequential data \cite{kidger2019deep,yang2022developing}, underlies signature kernel methods \cite{kiraly2019kernels,salvi2021signature}, and more recently, has been connected to random matrix theory through related randomised matrix valued path developments \cite{lou2023pcf,cirone2023neural,cass2026free,crew2025quantum}.

In the latter setting, one develops a path $\gamma$ into a matrix Lie group $G_N$ by solving a controlled differential equation 
\begin{equation}\label{eq:def-feature-Z}
\dd U = U\cdot M(\dd\gamma)
\end{equation}
with random $N\times N$ matrix coefficients. In coordinates 
\begin{equation}
M(\dd\gamma) = \sum_\mu A_\mu\,\dd\gamma^\mu
\end{equation}
for a vector of random matrices
$A_1,\ldots,A_d$. The resulting features $U(\gamma)\in G_N$ are known as the (path-)development. Taking a pair of these developments $U(\gamma)$ and $U(\sigma)$, with the same choice of $M$ but with different driving paths, induces a family of kernels on path space via the Hilbert-Schmidt inner product:
\begin{equation}
\kappa_{M}(\gamma,\sigma)=\tr\, (U(\gamma)U(\sigma)^{*}).
\end{equation}
These kernels can be used for discrimination between probability measures on paths \cite{lou2023pcf}, and is related to the theory of characteristic functions on path space \cite{chevyrev2016characteristic}. Taking the normalised matrix trace $\frac{1}{N}\tr\,U(\gamma)$---a scalar summary of the development, also known in the gauge theory context as a Wilson line---and averaging over the random matrices, one obtains a function $\frac{1}{N}\langle\tr U(\gamma;s,t)\rangle$ of the path $\gamma$, with endpoints $s\le t$. Its large-$N$ limit, the Wilson line observable
\begin{equation}\label{eq:wilson-intro}
  \cW(\gamma;s,t) := \lim_{N\to\infty}\tfrac{1}{N}\bigl\langle\tr U(\gamma;s,t)\bigr\rangle ,
\end{equation}
defined for every sub-interval $[s,t]$, is one of the main objects of study in this paper. The unitary path development in particular (where $G_N$ is taken to be the unitary group) can be implemented efficiently as a quantum circuit using sparse Pauli approximations to the Gaussian Unitary Ensemble (GUE), giving rise to a so-called quantum signature kernel \cite{crew2025quantum}. We therefore interchangeably refer to the unitary path development or the quantum path signature.

Random matrix theory and the theory of path signatures have developed largely independently. The randomised path development links them since it is both a matrix model observable and a feature on path space. This allows us to build a somewhat unexpected bridge between theoretical high energy physics and time series analysis. As we elucidate in this work, the geometric and algebraic machinery of matrix models, spectral curves, topological recursion, and the genus expansion \cite{t1993planar,migdal1983loop,eynard2016counting,eynard2007invariants}, can now be imported to the study of randomised path developments and the machine learning of sequential data. Conversely, the path signature provides new matrix model observables with an interesting path-space interpretation.

The large-$N$ limits of randomised path developments have been computed in two cases. For developments into $\mathrm{GL}(N;\mathbb{C})$ with Gaussian random matrices, Muça Cirone, Lemercier, and Salvi \cite{cirone2023neural} showed that the limiting kernel\footnote{Defined as the Hilbert-Schmidt inner product of two Gaussian path developments.} is the classical signature kernel \cite{kiraly2019kernels} $K_{\mathrm{sig}}(\gamma,\sigma) = \langle S(\gamma),S(\sigma)\rangle$, where $S(\gamma)$ is the path signature (as defined in Section~\ref{sec:setup}) and $\langle\cdot,\cdot\rangle$ is the inner product on the tensor algebra. For developments into $G_N = \mathrm{U}(N;\mathbb{C})$ with GUE matrices, Cass and Turner \cite{cass2026free} studied the large-$N$ limit of the normalised trace $\frac{1}{N}\langle\tr U(\gamma;s,t)\rangle$---the Wilson line or trace of the quantum path signature \eqref{eq:wilson-intro}. For both cases, Cass and Turner used tools from free probability \cite{voiculescu1986addition,nica2006lectures,anderson2010introduction} to establish a universality property in that the limit does not depend on the distribution beyond its second moment. They further showed that in the unitary case the Wilson line satisfies a quadratic functional equation, the planar path development equation
\begin{equation}\label{eq:CT-intro}
  \cW(\gamma;s,t) = 1 - \int_{s<u<r<t}
  \cW(\gamma;s,u)\,\cW(\gamma;u,r)\,
  \langle\dd\gamma(u),\dd\gamma(r)\rangle ,
\end{equation}
where $s<u<r<t$ are ordered times and $\dd\gamma(u),\dd\gamma(r)$ are the increments at $u,r$, paired by the Euclidean inner product. Their kernel on pairs of paths introduced in \cite{cass2026free} is built from this single-path Wilson line by concatenation. 

The planar path development equation is a \emph{loop equation}. In random matrix theory the large-$N$ limit is governed by the Schwinger--Dyson equations \cite{migdal1983loop,eynard2016counting,anderson2010introduction}. These are algebraic identities among the moments $\frac{1}{N}\langle\tr(A_{i_1}\cdots A_{i_n})\rangle$, where $A_{i_\ell}$ are random matrices with $i_\ell\in\{1,\cdots,d\}$, are the matrix-model analogue of the equations of motion in quantum field theory. The present work can be summarised by the statement that \eqref{eq:CT-intro} is the image of these equations under a particular non-commutative transform which we will introduce --- a Schwinger--Dyson equation on path space. When $d=1$ and $\gamma$ is a straight line, for instance, under a classical Laplace transform it reduces to the spectral curve of the GUE,
\begin{equation}\label{eq:semicircle-res-eq}
    R^2-zR+1=0,
\end{equation}
whose solution
\begin{equation}\label{eq:Gaussian-spectral-curve}
R(z)=(z-\sqrt{z^2-4})/2
\end{equation}
encodes the Wigner semicircle law \cite{wigner1993characteristic,anderson2010introduction}, and the Wilson line is recovered by the inverse Laplace transform as a contour integral:
\begin{equation}\label{eq:contour-intro}
  \cW(\gamma;0,\tau) = \frac{1}{2\pi i}\oint R(z)\,
 e^{iz\tau}\,\dd z = \int_{-2}^{2}e^{i y \tau}\, \frac{\sqrt{4-y^2}}{2\pi}\,\dd y.
\end{equation}
For $d\geq2$ and nonlinear paths, however, the classical Laplace transform no longer applies.

In this paper we make this relationship explicit for all $d$ and all paths. We introduce a non-commutative Laplace transform generalising \eqref{eq:contour-intro}, together with a convolution product on path space that the transform takes to a concatenation product in the tensor algebra. Under this correspondence the planar path development equation \eqref{eq:CT-intro} becomes a non-commutative spectral curve---an algebraic fixed-point equation in a free algebra---and is recovered from it by evaluation against the path signature: the transform thus maps between the algebraic geometry of the spectral curve and path space. Moreover, this geometric perspective enables generalisations of the planar path development equation that are not possible to obtain from the functional equation alone. In particular, the $1/N$ topological recursion expansion of random matrix theory translates, via the transform, into a sequence of linear Volterra integral equations on path space (Theorem \ref{thm:genus-g-loop}), providing systematic corrections to the planar limit ($N=\infty$).

\paragraph{Contributions.}
The main contributions of this work are summarised as
\begin{enumerate}
\item The definition of the non-commutative Laplace transform (Definition~\ref{def:nc-laplace}), its inverse via evaluation against the signature (Proposition~\ref{prop:inversion}), and the product--convolution property (Proposition~\ref{prop:ev-props}).

  \item The recasting, under the non-commutative Laplace transform, of the planar path development equation of Cass and Turner \cite{cass2026free} as the spectral form of the Schwinger--Dyson equations of free probability (Theorems~\ref{thm:nc-spectral} and~\ref{thm:CT}).
\item The genus-$g$ loop equations on path space 
  (Theorem \ref{thm:genus-g-loop}). Integral equations that generalise the planar path development equation of Cass--Turner to include $1/N$ corrections.
\end{enumerate}

\paragraph{Outline.}
Section~\ref{sec:setup} fixes notation and recalls the definitions of the path signature and the randomised path development. Section~\ref{sec:nc-laplace} then introduces the non-commutative Laplace transform, its inverse via evaluation against the signature, and the product--convolution property. Section~\ref{sec:spectral} applies the transform to the Schwinger--Dyson equations of the GUE: we derive the non-commutative spectral equation, recover the planar path development equation and its pair-partition solution, and connect to the classical spectral curve. Section~\ref{sec:genus} extends the construction to all orders in $1/N$, deriving the genus-$g$ loop equation on path space and verifying the genus one correction numerically. We conclude in Section~\ref{sec:discussion} with a discussion of applications and directions for future work.

\section{Preliminaries}\label{sec:setup}

We begin with a brief introduction to path signatures and randomised path developments, fixing notation along the way. We refer the reader to \cite{cass2024lecture,chevyrev2025primer} for recent more comprehensive treatments with proofs.

\paragraph{Notation.}

We work in Euclidean $\mathbb{R}^d$ with the standard inner product $\delta_{\mu\nu}$. Spatial indices $\mu,\nu,i_k\in\{1,\ldots,d\}$ are placed as superscripts on path components ($\gamma^\mu$, $\dd\gamma^\mu$, $S^{i_1\cdots i_n}$) and as subscripts on algebraic objects ($A_\mu$, $\xh_\mu$, $M_{i_1\cdots i_n}$); repeated indices in either position are summed using $\delta_{\mu\nu}$. We also take $I= i_1\cdots i_n$ as a collection of $n$ indices, a word, and use the shortened representations $c^{I}:=c^{i_1\cdots i_n}$, $e_I:=e_{i_1}\otimes\cdots\otimes e_{i_n}$ and the noncommutative product $\xh_I:=\xh_{i_1} \cdots \xh_{i_n}$. The length of a word $I$ is denoted by $|I|$.

\subsection{The path signature}\label{sec:path-sig}

Let $\gamma\colon[0,T]\to\mathbb{R}^d$ be a smooth path.  Its signature is the formal tensor series
\begin{equation}\label{eq:sig}
  S(\gamma;s,t) = \sum_{n=0}^{\infty}
  S^{i_1\cdots i_n}_{s,t}(\gamma)\;
  e_{i_1}\otimes\cdots\otimes e_{i_n}
  \;\in\; T((\mathbb{R}^d)) 
  := \prod_{n\geq 0}(\mathbb{R}^d)^{\otimes n},
\end{equation}
where $0\le s\le t\le T$, $S^{\emptyset}=1$, and $\{e_1,\ldots,e_d\}$ is the standard basis of $\mathbb{R}^d$. The degree $n$ component of the signature is the iterated integral
\begin{equation}\label{eq:sig-comp}
  S^{i_1\cdots i_n}_{s,t}(\gamma) = 
  \int_{s<t_1<\cdots<t_n<t}
    \dd\gamma^{i_1}(t_1)\cdots\dd\gamma^{i_n}(t_n)\,.
\end{equation}
The path signature is the solution to the controlled differential equation
\begin{equation}\label{eq:ODE-path-sig}
\dd S = S\otimes\dd\gamma\quad;\qquad S(\gamma;s,s)=1,
\end{equation}
in the tensor algebra and we equivalently write the solution as the path-ordered exponential $S = \cP\exp\!\int_s^t\dd\gamma$. 

We now recall three standard properties of the path signature that will be used in the following \cite{lyons2007differential,reutenauer2003free}:

\begin{itemize}
    \item \emph{Factorial decay}:\footnote{The $\|\gamma\|_1=\sup_{\mathcal{D}}\sum_k|\gamma(t_{k+1})-\gamma(t_k)|$ is the $1$-variation, where the supremum is taken over partitions $\mathcal{D}=\{s=t_0<\cdots<t_m=t\}$.} $|S^{i_1\cdots i_n}_{s,t}|\leq\|\gamma\|_1^n/n!$ 
    \item \emph{Chen's identity}: $S_{s,t} = S_{s,u}\otimes S_{u,t}$ for $s<u<t$\,.
    \item \emph{Shuffle product}: $S^I\cdot S^J = S^{I\, \shuffle\, J}$, where the shuffle $I\shuffle J$ is the (formal) sum over all
    interleavings of the words $I$ and $J$ that preserve the order of the letters within each word.\footnote{For instance $1\shuffle 23 = 123+213+231$, so that $S^1 S^{23} = S^{123}+S^{213}+S^{231}$.}
\end{itemize}

\subsection{Quantum path signature}

Fix $d$ independent $N\times N$ GUE matrices $A_1,\ldots,A_d$ with joint law $\propto\exp(-\frac{N}{2}\sum_\mu\tr A_\mu^2)$.  The \emph{unitary path development} \cite{lou2023pcf,chevyrev2016characteristic}  or \emph{quantum path signature} \cite{crew2025quantum} $U(\gamma;s,t)$ in $\mathrm{U}(N,\mathbb{C})$ is defined as the solution to the matrix valued ODE
\begin{equation}\label{eq:dev}
  \dd U = i U A_\mu\,\dd\gamma^\mu,
  \qquad U(\gamma;s,s)=\mathbf{1}_N.
\end{equation}
Since the generators $iA_\mu$ are anti-Hermitian, $U$ is unitary, hence the terminology \emph{quantum}. The path development can be expanded against the signature:
\begin{equation}\label{eq:U-expand}
    U(\gamma;s,t) = \sum_{I} S^{I}_{s,t}(\gamma)\,(iA)_I, \qquad (iA)_I := (iA_{i_1})\cdots(iA_{i_n}) = i^{|I|}A_I.
\end{equation}
which follows from differentiating and using the signature ODE \eqref{eq:ODE-path-sig}. Taking the normalised trace and the large-$N$ limit we then define the \emph{planar Wilson line} as an expectation value:
\begin{equation}\label{eq:wilson}
  \cW(\gamma;s,t) 
  = \lim_{N\to\infty}\frac{1}{N}
  \bigl\langle\tr\,U(\gamma;s,t)\bigr\rangle
  = \sum_{I}S^{I}_{s,t}(\gamma)\,M_{I},
\end{equation}
where the \emph{planar moments} are \begin{equation}\label{eq:planar-moments}
 M_I := \lim_{N\to\infty}\frac{1}{N}\bigl\langle\tr\,(iA)_I\bigr\rangle  = i^{|I|}\lim_{N\to\infty}\frac{1}{N}\bigl\langle\tr A_I\bigr\rangle,
  \qquad M_\emptyset = 1.
\end{equation}

Using the genus expansion \cite{t1993planar,anderson2010introduction}, one can show that this limit exists and equals the sum over planar (genus-$0$) Wick contractions.\footnote{The genus enters through 't Hooft's double line representation of Wick contractions \cite{t1993planar}: each contraction of a trace of random matrices is drawn as a ribbon, and a set of contractions tiles an orientable surface whose genus $g$ weights it by $N^{-2g}$. The planar ($g=0$) terms dominate as $N\to\infty$. The higher genus corrections are the subject of Section~\ref{sec:genus} and we review these ideas more precisely in definition \ref{def:genus-pp}.}  The interchange of the limit and the sum in \eqref{eq:wilson} is justified by the factorial decay of the signature and the bound on the growth of the planar moments ($|M_{i_1\cdots i_n}|\leq 4^{n/2}$); see \cite{cass2026free} for details.

\section{The non-commutative Laplace transform}\label{sec:nc-laplace}

In this Section we define the non-commutative Laplace transform and its inverse in terms of the tensor algebra and the path signature. We demonstrate several useful properties which we later apply to random matrix models in Section~\ref{sec:spectral}.

\subsection{Definition}

Let $\xh_1,\ldots,\xh_d$ be non-commuting indeterminate variables. We write $\mathbb{C}\langle\langle\xh_1,\ldots,\xh_d\rangle\rangle$ for the algebra of formal power series in the $\xh_\mu$, i.e.\ formal sums $\sum_I c^{I}\,\xh_I$ graded by word length, with the concatenation product. Identities in this Section are to be understood degree-by-degree: each homogeneous component is a finite sum, and equality means equality of every homogeneous component. 

\begin{definition}[Evaluation map]\label{def:ev}
The \emph{evaluation map}
$\ev\colon\mathbb{C}\langle\langle\xh\rangle\rangle\to\mathcal{S}$ is the linear map sending each word to the corresponding signature functional,
\begin{equation}\label{eq:ev}
  \ev\colon\xh_I\longmapsto S^{I},
  \qquad \ev(1)=1,
\end{equation}
extended linearly to a series $\Omega=\sum_I c^I\xh_I$ by $\ev(\Omega)=\sum_I c^I S^I$. Here $S^I$ denotes the path functional $(\gamma;s,t)\mapsto S^{I}_{s,t}(\gamma)$, and $\mathcal{S}$ is the space of path functionals of the form $F(\gamma;s,t)=\sum_{I}c^{I}\,S^{I}_{s,t}(\gamma)$ whose coefficients have at most exponential growth, i.e.\ $|c^I|\leq C\lambda^{|I|}$ for some constants $C,\lambda>0$. For a fixed path $\gamma$ we write $\ev_\gamma(\Omega):=\ev(\Omega)(\gamma;s,t)$ for the value of this functional at $\gamma$.
\end{definition}

\begin{proposition}[Coefficient correspondence]\label{prop:inversion}
On series of at most exponential growth---those $\Omega=\sum_I c^I\xh_I$ with $|c^I|\leq C\lambda^{|I|}$ for some $C,\lambda>0$---the evaluation map is a linear isomorphism onto $\mathcal{S}$:
\begin{enumerate}
\item the exponential growth of the coefficients and the factorial decay of the signature ensure $\ev(\Omega)=\sum_I c^I S^I$ converges absolutely on every path of finite $1$-variation, so $\ev$ is well defined;
\item conversely, every $F\in\mathcal{S}$ determines its coefficients $c^I$ uniquely: the iterated integrals $\{S^I\}$, indexed by words $I$, are linearly independent as functionals on path space
\cite{chen1957integration,reutenauer2003free,hambly2010uniqueness}.
\end{enumerate}
\end{proposition}

\begin{definition}[Non-commutative Laplace transform]\label{def:nc-laplace}
The \emph{non-commutative bare Laplace transform} $\cL_B\colon\mathcal{S}\to\mathbb{C}\langle\langle\xh_1,\ldots,\xh_d\rangle\rangle$ is the inverse of the evaluation map, i.e. 
\begin{equation}
  \cL_B\colon S^{I}\longmapsto\xh_I ,
\end{equation}
extended linearly. For each $\mu\in\{1,\ldots,d\}$ the \emph{$\mu$-directional non-commutative Laplace transform} is the degree-raising map
\begin{equation}\label{eq:nc-laplace}
  \cL_\mu(F) := \cL_B(F)\,\xh_\mu ,
\end{equation}
appending the letter $\xh_\mu$ on the right.
\end{definition}

The transform is elementary. $\cL_B$ simply records the coefficients of a functional in the basis of iterated integrals, i.e. it is a non-commutative generating function. The power lies in the interpretation, and it enters through the simple operation of appending a letter. It is not the bare $\cL_B$ but the directional transform $\cL_\mu = \cL_B(\,\cdot\,)\,\xh_\mu$ that corresponds to the familiar Laplace theory. The appended letter implements integration against $\dd\gamma^\mu$ (Proposition~\ref{prop:ev-props}), giving a convolution theorem and a Laplace transform calculus that reduces to the classical Laplace transform for $d=1$ (Example~\ref{ex:classical}). As we will see in the following, this converts the quadratic loop equations of path developments into algebra, just as the classical transform converts differential
equations into polynomial ones.

\begin{example}\label{ex:classical}
For $d=1$ and the straight line $\gamma(t)=t$ on $[0,\tau]$, the signature components are $S^{n}_{0,\tau}=\tau^n/n!$, so the basis functionals of $\mathcal S$ are the monomials $\tau^n/n!$. Their classical Laplace transforms are
\begin{equation*}
  \int_0^\infty e^{-z\tau}\,\frac{\tau^n}{n!}\,\dd\tau
  = \frac{1}{z^{\,n+1}}.
\end{equation*}
Identifying $\xh$ with $1/z$, the right-hand side is the word $\xh^{\,n+1}$, which is exactly the directional transform $\cL_\mu(S_n)=\cL_B(S_n)\,\xh=\xh^{\,n+1}$. It is therefore $\cL_\mu$, and not $\cL_B$, that restricts to the classical Laplace transform. The appended letter plays the role of the degree shift by $z^{-1}$ in the classical formula. In the reverse direction, $\ev_\gamma$ sends $\xh^{\,n}\mapsto\tau^n/n!$, dividing the $n$-th coefficient by $n!$. The inverse transform is a Borel summation.
\end{example}

Applying the bare transform to the Wilson line \eqref{eq:wilson} defines the \emph{moment generating function} and the \emph{non-commutative directional resolvent}:
\begin{equation}\label{eq:R-nc}
  W := \cL_B[\cW] = \sum_{I} M_{I}\,\xh_I \;\in\;\mathbb{C}\langle\langle\xh_1,\ldots,\xh_d\rangle\rangle,\qquad
  R_\mu := \cL_\mu[\cW] = W\,\xh_\mu .
\end{equation}
For $d=1$, under the identification $\xh=i/z$ one has $W_{\mathrm{ab}}(i/z)=z\,R(z)$, with $R$ the resolvent of the semicircle law  \eqref{eq:Gaussian-spectral-curve}.

\subsection{The inverse transform}\label{sec:inverse}

By construction the transform pair inverts: $\ev\circ\cL_B=\mathrm{id}_{\mathcal S}$ and, on series of at most exponential growth, $\cL_B\circ\ev=\mathrm{id}$, both maps being mutually inverse on the basis ($S^I\mapsto\xh_I\mapsto S^I$) and extending by Proposition~\ref{prop:inversion}. In particular the planar Wilson line is recovered from the moment generating function \eqref{eq:R-nc} by evaluation,
\begin{equation}\label{eq:inverse}
  \cW(\gamma;s,t)=\ev_\gamma(W).
\end{equation}
This inversion is a duality pairing. Define the signature element
\begin{equation}\label{eq:E-gamma}
  E_\gamma = \sum_{I}S^{I}_{s,t}(\gamma)\,\xh_I \;\in\;\mathbb{C}\langle\langle\xh\rangle\rangle,
\end{equation}
the unique solution of $\dd E = E\,\xh_\mu\,\dd\gamma^\mu$, $E(s)=1$: the
signature of $\gamma$ viewed as an element of the tensor algebra. It is
dual to $\ev_\gamma$ in the sense that
\begin{equation}\label{eq:pairing}
  \ev_\gamma(\Omega) = \langle \Omega, E_\gamma\rangle
  \qquad\text{for all }\Omega,
  \qquad \langle\xh_I,\xh_J\rangle = \delta_{IJ}.
\end{equation}
Thus $E_\gamma$ plays the role of the classical Laplace kernel and the word
pairing that of the contour integral.
For $d=1$, taking $\xh=z\in\mathbb{C}$, $E_\gamma = e^{z\,\Delta\gamma}$ with $\Delta\gamma=\gamma(t)-\gamma(s)$, and $\langle W,E_\gamma\rangle=\sum_n M_n(\Delta\gamma)^n/n!$ is exactly the series evaluated by $\frac{1}{2\pi i}\oint R(z)\,e^{iz\Delta\gamma}\,\dd z$, matching \eqref{eq:contour-intro}; see Example~\ref{ex:classical}.

For any non-straight path in $d\geq 2$, the word pairing $\langle W,E_\gamma\rangle$ is then the natural non-commutative analogue of the contour integral form of the inverse Laplace transform.

\subsection{Properties of the transform pair}\label{sec:props}

The evaluation map $\ev_\gamma$ and directional Laplace transform also interact with the algebraic structure of $\mathbb{C}\langle\langle\xh\rangle\rangle$ in ways that parallel classical Laplace transform theory. In this Section, we first define a generalised path version of convolution and then present the main properties of the non-commutative Laplace transform.

\begin{definition}[Path convolution]\label{def:path-conv}
For functions $f(\gamma;s,t)$ and $g(\gamma;s,t)$ defined for $s\leq t$, the \emph{path convolution} is
\begin{equation}\label{eq:path-conv}
  (f\star_\gamma g)(s,t)  = \int_s^t\!\int_s^r  f(\gamma;s,u)\,g(\gamma;u,r)\, \langle\dd\gamma(u),\dd\gamma(r)\rangle.
\end{equation}
\end{definition}

\begin{theorem}[Properties of the transform]\label{prop:ev-props}
Let $f,g\in\mathcal S$ and let $\gamma\colon[s,t]\to\mathbb{R}^d$ be smooth. Then:
\begin{enumerate}
\item \textbf{Integration}:
  $\ev_\gamma\bigl(\cL_\mu(f)\bigr)
  =\int_s^t f(\gamma;s,u)\,\dd\gamma^\mu(u)$.
\item \textbf{Product--convolution}: $\ev_\gamma\Bigl(\sum_{\mu=1}^d\cL_\mu(f)\,\cL_\mu(g)\Bigr)  = f\star_\gamma g$.
\item \textbf{Shuffle--product}: $\ev_\gamma\bigl(\cL_B(f)\shuffle\cL_B(g)\bigr)=f\cdot g$.
\end{enumerate}
\end{theorem}
\begin{proof}
By linearity it suffices to prove each identity on basis words, with $\cL(f)=\sum_K c^K\xh_K$, $\cL(g)=\sum_L d^L\xh_L$.
\begin{enumerate}
\item Since $\cL_\mu(f)=\cL(f)\,\xh_\mu$, it suffices to show $\ev_\gamma(\xh_I\,\xh_\mu)=\int_s^t S^I_{s,u}\,\dd\gamma^\mu(u)$: the simplex integral for $S^{I\mu}_{s,t}$, conditioned on its last coordinate $u:=t_n$, factorises as $S^{I\mu}_{s,t}=\int_s^t S^{I}_{s,u}\,\dd\gamma^\mu(u)$.
\item Since $\cL_\mu(f)\,\cL_\mu(g)=\cL_B(f)\,\xh_\mu\,\cL_B(g)\,\xh_\mu$, it suffices to prove the word identity
\begin{equation}\label{eq:marked}
  S^{K\mu L\mu}_{s,t}= \int_{s<u<r<t} S^{K}_{s,u}\,S^{L}_{u,r}\,\dd\gamma^{\mu}(u)\,\dd\gamma^{\mu}(r)\,.
\end{equation}
By definition \eqref{eq:sig-comp}, $S^{K\mu L\mu}_{s,t}$ is the integral over $s<t_1<\cdots<t_n<t$ ($n=|K|+|L|+2$), in which the two appended letters $\mu$ correspond to the coordinates $u:=t_{|K|+1}$ and $r:=t_n$. Integrating over the remaining coordinates at fixed $(u,r)$, the ordering constraint decomposes: the coordinates carrying $K$ range over $s<t_1<\cdots<t_{|K|}<u$, giving $S^{K}_{s,u}$; those carrying $L$ range over $u<\cdots<r$, giving $S^{L}_{u,r}$; and the marked pair ranges over $s<u<r<t$. This proves \eqref{eq:marked}. Summing over $K,L$ and over $\mu$, which contracts the marked increments into $\langle\dd\gamma(u),\dd\gamma(r)\rangle$, gives (iii).
\item At the level of words the statement is a rephrasing of the shuffle identity of Section~\ref{sec:setup}, $\ev_\gamma(\xh_I\shuffle\xh_J)=S^{I\shuffle J}_{s,t} =S^I_{s,t}\,S^J_{s,t}$.

\end{enumerate}
\end{proof}

\begin{example}[Classical Laplace transform]\label{ex:classical-reduction} For $d=1$ and the straight line $\gamma(t)=t$ on $[0,\tau]$, every functional in $\mathcal S$ depends only on the interval length, $f(s,t)=f(t-s)$, and the path convolution collapses to the integral of the classical convolution:
\begin{equation}
  (f\star_\gamma g)(0,\tau)  = \int_0^\tau\!\!\int_0^r f(u)\,g(r-u)\,\dd u\,\dd r  = \int_0^\tau (f*g)(r)\,\dd r .
\end{equation}
Since $\cL_\mu$ reduces to the classical Laplace transform on the monomial basis (Example~\ref{ex:classical}), properties (ii) and (iii) reduce to the classical rules $\cL[\int_0^t f]=z^{-1}\hat f$ and $\cL[f*g]=\hat f\,\hat g$: the transform of the path convolution is $z^{-1}\hat f\hat g$. On monomials, the classical beta function identity directly gives 
 \begin{equation}
 \frac{\tau^m}{m!}\star_\gamma\frac{\tau^n}{n!}  =\frac{\tau^{m+n+2}}{(m+n+2)!}\,.
 \end{equation}
\end{example}

Property (ii) is the analogue of the classical $\cL^{-1}[z^{-1}\hat f]=\int_0^t f$ where $\mathcal{L}$ is the classical Laplace transform and $\hat{f}=\cL[f]$. Property (iii) is the convolution theorem, the analogue of $\cL^{-1}[\hat f\,\hat g]=f*g$, with the product of directional transforms summed over directions; property (iv) is its dual, with the shuffle playing the role of the complex convolution of classical Laplace theory.

The $\mu$-directional Laplace transform $\cL_\mu$ is an isomorphism of $\mathcal S$ onto $\mathbb{C}\langle\langle\xh\rangle\rangle\,\xh_\mu$, the words ending in $\mu$ so that no directional transform reaches the empty word. Correspondingly, the path convolution has no unit. Classically the missing unit is restored by including distributions ($\cL[\delta]=1$). The constant functional $S^{\emptyset} = 1$ is instead the unit of the pointwise product---indeed $\cL_\mu(1)=\xh_\mu$, not the empty word.

\section{Application to randomised path developments}\label{sec:spectral}

We now apply the non-commutative Laplace transform to randomised matrix valued path developments using the Schwinger--Dyson equations of random matrix theory. We obtain a non-commutative spectral equation governing the development, and derive the planar path development equation \eqref{eq:CT-intro} as a corollary.

The central algebraic identity, the planar Schwinger--Dyson equation, is a recursion equation for the planar moments \eqref{eq:planar-moments}:
\begin{equation}\label{eq:SD}
  M_{i_1\cdots i_n\,\mu} = -\sum_{k=1}^{n}\delta_{i_k,\mu}\,M_{i_1\cdots i_{k-1}}\,M_{i_{k+1}\cdots i_n}\,.
\end{equation}
This equation can be obtain by integration by parts in the matrix integral
\begin{equation}\label{eq:matrix-int-A-mu}
    \langle\tr(A_{i_1}\cdots A_{i_n}A_\mu)\rangle\,,
\end{equation}
where the average $\langle\,\cdot\,\rangle$ is the GUE matrix integral with the Gaussian potential of the law fixed in Section~\ref{sec:setup}:
\begin{equation}\label{eq:GUE-average}
\langle f\rangle = \frac{1}{Z}\int f(A)\,e^{-N\tr V(A)}\,\prod_\mu\dd A_\mu\,,\quad V(A)=\frac{1}{2}\sum_\nu A_\nu^2\,,
\end{equation}
$\dd A_\mu$ is the flat measure on Hermitian matrices, and $Z$ is the normalisation. Differentiating \eqref{eq:matrix-int-A-mu} with respect to $(A_\mu)_{ij}$ and using $\langle\partial/\partial A_{ij}[\cdots]\rangle=0$ gives, in the large-$N$ limit where the two-trace expectation factorises into a product of moments, the Schwinger–Dyson recursion \eqref{eq:SD}.

\subsection{The non-commutative spectral equation}\label{sec:nc-spectral}

In the following, for notational convenience, we define  
\begin{equation}\label{eq:def-Phi}
\Phi(A) = \sum_{\mu=1}^d \xh_\mu\,A\,\xh_\mu\,.
\end{equation}
The operator $\Phi$ relates the bare and directional transforms: for $f,g\in\mathcal S$,
\begin{equation*}
  \cL_B(f)\cdot\Phi\bigl(\cL_B(g)\bigr)
  = \sum_{\mu=1}^d \cL_\mu(f)\,\cL_\mu(g),
\end{equation*}
turning a product of bare transforms into the summed product of directional ones, whose evaluation is the path convolution $f\star_\gamma g$ (Proposition~\ref{prop:ev-props}(iii)).

We start with a self-contained treatment of the well-known non-commutative spectral curve \cite{anderson2010introduction} in the theorem below.

\begin{theorem}[Non-commutative spectral equation]\label{thm:nc-spectral}
The moment generating function $W$ of equation \eqref{eq:R-nc} is the unique solution in $\mathbb{C}\langle\langle\xh_1,\ldots,\xh_d\rangle\rangle$ with $W|_{\xh=0}=1$ of the fixed point equation
\begin{equation}\label{eq:nc-spectral}
  W = 1 - W\cdot\Phi(W) = 1 - \sum_{\mu=1}^d W\,\xh_\mu\,W\,\xh_\mu.
\end{equation}
\end{theorem}

\begin{proof}
\emph{Existence.}
Every nonempty word may be factored uniquely as $I=J\mu$. Grouping the terms of $W$ by their
final letter and factoring the trailing $\xh_\mu$ to the right we may write
\begin{equation}\label{eq:W-peel}
  W = M_{\emptyset} + \sum_{\mu=1}^d\Bigl(\sum_J M_{J\mu}\,\xh_J\Bigr)\xh_\mu
    = 1 + \sum_{\mu=1}^d G_\mu\,\xh_\mu,
  \qquad
  G_\mu := \sum_J M_{J\mu}\,\xh_J,
\end{equation}
where $G_\mu$ generates the moments whose last letter is $\mu$. We show $G_\mu = -W\xh_\mu W$. Contracting the Schwinger-Dyson equation \eqref{eq:SD} with $\xh_{i_1}\cdots\xh_{i_n}$ and summing over all words of all lengths $n\geq1$,
\begin{equation}\label{eq:Gmu-step1}
  G_\mu = -\sum_{n\geq 1}\sum_{k=1}^n\delta_{i_k,\mu}\,M_{i_1\cdots i_{k-1}}M_{i_{k+1}\cdots i_n}\,\xh_{i_1}\cdots\xh_{i_n}.
\end{equation}
The Kronecker delta fixes position $k$ to carry colour $\mu$, splitting the monomial into a left block $\xh_{i_1}\cdots\xh_{i_{k-1}}$, the middle $\xh_\mu$, and a right block $\xh_{i_{k+1}}\cdots\xh_{i_n}$; since the $\xh$'s do not commute, the blocks retain their positions. Setting $p=k-1\geq0$, $q=n-k\geq0$ and summing over all $n$ lets $p,q$ range independently, so
\begin{equation}\label{eq:Gmu-nc}
  G_\mu = -\Bigl(\sum_{p\geq0}M_{j_1\cdots j_p}\xh_{j_1}\cdots\xh_{j_p}\Bigr) \,\xh_\mu\,\Bigl(\sum_{q\geq0}M_{l_1\cdots l_q}\xh_{l_1}\cdots\xh_{l_q}\Bigr) = -W\,\xh_\mu\,W.
\end{equation}
Substituting into \eqref{eq:W-peel} gives $W = 1 - \sum_\mu W\xh_\mu W\xh_\mu = 1-W\cdot\Phi(W)$.

\emph{Uniqueness.} The right-hand side of \eqref{eq:nc-spectral} is raises degree, the two appended letters in $W\cdot\Phi(W)=\sum_\mu W\xh_\mu W\xh_\mu$ mean its degree-$n$ component $W_{(n)}$ depends only on the components of $W$ of degree $\leq n-2$. Equating homogeneous components of \eqref{eq:nc-spectral} thus determines $W_{(n)}$ from $W_{(0)},\dots,W_{(n-2)}$; by induction on $n$, starting from the constant term $W_{(0)}=1$, the solution is unique.
\end{proof}

We can also recover the Gaussian spectral curve \eqref{eq:semicircle-res-eq} from a process of abelianisation. Abelianising, $\xh_\mu\mapsto x_\mu$ with the $x_\mu$ commuting, the quadratic term collapses, $W\xh_\mu W\xh_\mu\mapsto x_\mu^2 W_{\mathrm{ab}}^2$, so $W\cdot\Phi(W)\mapsto\|\mathbf{x}\|^2W_{\mathrm{ab}}^2$ and \eqref{eq:nc-spectral} becomes the Gaussian spectral curve 
\begin{equation}\label{eq:curve}
  \|\mathbf{x}\|^2\,W_{\mathrm{ab}}^2 + W_{\mathrm{ab}} - 1 = 0,   \qquad \|\mathbf{x}\|^2=\textstyle\sum_\mu x_\mu^2,
\end{equation}
whose branch with $W_{\mathrm{ab}}(\mathbf 0)=1$ is
\begin{equation}
 W_{\mathrm{ab}}(\mathbf{x}) = (\sqrt{1+4\|\mathbf{x}\|^2}-1)/(2\|\mathbf{x}\|^2) = \sum_{n\geq0}(-1)^nC_n\|\mathbf{x}\|^{2n}\,,
\end{equation}
the alternating Catalan generating function ($C_n=\tfrac{1}{n+1}\binom{2n}{n}$). The curve depends only on $\|\mathbf{x}\|^2$ due to the $O(d)$ invariance of the $d$ i.i.d. GUE matrices.

\begin{example}\label{ex:picard}
We compute the first few orders by substituting the graded ansatz $W=\sum_n W_{(n)}$ into \eqref{eq:nc-spectral} and collecting equal degrees (words of the same length). Since $W\cdot\Phi(W)$ raises degree by two,
\begin{equation}
  W_{(n)} = -\sum_{p+q=n-2}\sum_{\mu} W_{(p)}\,\xh_\mu\,W_{(q)}\,\xh_\mu,
\end{equation}
then from $W_{(0)}=1$ each component follows from lower ones:
\begin{align}
  W_{(2)} &= -\sum_\mu \xh_\mu^2, \\
  W_{(4)} &= -\sum_{\mu}\bigl(\xh_\mu\,W_{(2)}\,\xh_\mu        + W_{(2)}\,\xh_\mu\,\xh_\mu\bigr) = \sum_{\mu,\nu}\bigl(\xh_\mu\xh_\nu^2\xh_\mu    + \xh_\nu^2\xh_\mu^2\bigr),
\end{align}
and all odd components vanish. 
\end{example}
The words appearing in $W$ have a combinatorial description, which we discuss further in Section~\ref{subsec:pp}: a \emph{pair partition} of $\{1,\ldots,2n\}$ groups the positions into $n$ disjoint pairs; it is \emph{non-crossing} if no two pairs $(a,b)$, $(c,d)$ interlace as $a<c<b<d$, and it is \emph{colour-matched} for a word if paired positions carry the same letter. As we will see below, $W$ all pairs appearing are non-crossing and colour-matched as can been seen in the example above: the word $\xh_\mu\xh_\nu^2\xh_\mu$ (for $\mu\neq\nu$) admits the colour-matched non-crossing pairing $\{(1,4),(2,3)\}$; the word $\xh_1\xh_2\xh_1\xh_2$ does not appear, since the only colour-matched pairing of $1212$ is $\{(1,3),(2,4)\}$, which crosses.

\subsection{The planar path development equation}\label{sec:CT}

The planar path development equation \eqref{eq:CT-intro} of Cass--Turner \cite{cass2026free} is now an immediate corollary of the product convolution property of proposition \ref{prop:ev-props}:

\begin{theorem}\label{thm:CT}
For any smooth path $\gamma\colon[0,T]\to\mathbb{R}^d$, applying $\ev_\gamma$ to \eqref{eq:nc-spectral} yields
\begin{equation}\label{eq:CT}
  \cW(\gamma;s,t) = 1 - \int_s^t\!\int_s^r
  \cW(\gamma;s,u)\,\cW(\gamma;u,r)\,
  \langle\dd\gamma(u),\dd\gamma(r)\rangle
  = 1 - \cW\star_\gamma\cW.
\end{equation}
\end{theorem}
\begin{proof}
The coefficients of $W$ have exponential growth, $|M_I|\leq 4^{|I|/2}$ (Section~\ref{sec:setup}), and from \eqref{eq:inverse} $\cW=\ev_\gamma(W)\in\mathcal S$ and $W=\cL_B[\cW]$. The quadratic term of \eqref{eq:nc-spectral} is
\begin{equation}
\sum_\mu W\xh_\mu W\xh_\mu=\sum_\mu\cL_\mu(\cW)\,\cL_\mu(\cW)\,,
\end{equation}
so applying $\ev_\gamma$ and using Proposition~\ref{prop:ev-props}(i) and (iii) with $f=g=\cW$ gives $\cW = 1 - \cW\star_\gamma\cW$.
\end{proof}

\begin{remark}[Uniqueness on path space]\label{rem:pathspace-uniqueness}
Theorem \ref{thm:CT} shows that $\cW$ satisfies \eqref{eq:CT}. The algebraic uniqueness of Theorem \ref{thm:nc-spectral} does not transfer directly, since $\ev_\gamma$ is not injective. Whether \eqref{eq:CT} determines $\cW$ uniquely is a separate, analytic question, established for smooth (and bounded variation) paths by Cass and Turner \cite{cass2026free}. The rough path case requires further analysis, which we do not pursue here.
\end{remark}

\begin{example}\label{ex:d1}
For $d=1$ the letters commute and \eqref{eq:curve} reduces to a single scalar equation. Setting $x=i/z$ and $R(z)=z^{-1}W_{\mathrm{ab}}(i/z)$ gives the algebraic curve \eqref{eq:semicircle-res-eq}, solved by the resolvent \eqref{eq:Gaussian-spectral-curve}, the transform of the Wigner semicircle law $\rho(y)=\tfrac{1}{2\pi}\sqrt{4-y^2}$, $|y|\leq2$. Correspondingly the planar path development equation \eqref{eq:CT} becomes the classical scalar loop equation, and the Wilson line is the contour integral \eqref{eq:contour-intro} of the introduction.
\end{example}

\subsection{The pair-partition solution}\label{subsec:pp}

It is well-known that the ordered moments \eqref{eq:planar-moments} (the free semicircular family) are counted by non-crossing pair partitions \cite{nica2006lectures,anderson2010introduction}. In our present setting, this is rephrased as an explicit combinatorial solution to \eqref{eq:nc-spectral}:

\begin{theorem}\label{thm:nc2}
The unique solution of \eqref{eq:nc-spectral} is
\begin{equation}\label{eq:W-nc-sol}
  W = \sum_{n=0}^{\infty}(-1)^n\sum_{\pi\in\NC(2n)}
  \prod_{(a,b)\in\pi}\delta_{i_a i_b}\;\xh_{i_1}\cdots\xh_{i_{2n}},
\end{equation}
where $\NC(2n)$ is the set of non-crossing pair partitions of $\{1,\ldots,2n\}$. Equivalently, the ordered moment $M_{i_1\cdots i_{2n}}$ \eqref{eq:planar-moments} counts the non-crossing pairings compatible with the colouring. Odd moments are zero for the GUE.
\end{theorem}

Evaluating the solution on a path then gives a corresponding closed form combinatorial solution for the Wilson line.
\begin{theorem}\label{thm:complete}
For any smooth path $\gamma\colon[0,T]\to\mathbb{R}^d$, the Wilson line is the evaluation of the solution \eqref{eq:W-nc-sol}, given explicitly by
\begin{equation}\label{eq:W-explicit}
  \cW(\gamma;s,t)  = \ev_\gamma(W) = \sum_{n=0}^{\infty}(-1)^n\sum_{\pi\in\NC(2n)} \int_{\Delta_{2n}[s,t]} \prod_{(a,b)\in\pi}\langle\dd\gamma(t_a),\dd\gamma(t_b)\rangle,
\end{equation}
where $\Delta_{2n}[s,t]=\{s<t_1<\cdots<t_{2n}<t\}$. 
\end{theorem}

\begin{proof}
Apply $\ev_\gamma$ to \eqref{eq:W-nc-sol}: the monomial $\xh_{i_1}\cdots\xh_{i_{2n}}$ maps to $S^{i_1\cdots i_{2n}}_{s,t}(\gamma)$, and, using \eqref{eq:sig-comp}, the deltas contract the $2n$ colour indices into $n$ matched pairs, each giving $\langle\dd\gamma(t_a),\dd\gamma(t_b)\rangle$. 
\end{proof}

For completeness, we recover the well-known (see, for example \cite{cass2026free}) result for the $d=1$ planar Wilson line:

\begin{corollary}[Straight-line paths]\label{cor:bessel}
For $\gamma(t)=tv$ with $v\in\mathbb{R}^d$ on $[0,\tau]$,
\begin{equation}\label{eq:bessel}
  \cW(\gamma;0,\tau) = \sum_{n=0}^{\infty}(-1)^n C_n \frac{(\tau\|v\|)^{2n}}{(2n)!} = \frac{J_1(2\tau\|v\|)}{\tau\|v\|},
\end{equation}
with $J_1$ the Bessel function of the first kind.
\end{corollary}

\begin{proof}
The signature of $\gamma(t)=tv$ is $S^{i_1\cdots i_n}_{0,\tau}=\tau^n v^{i_1}\cdots v^{i_n}/n!$. Substituting into \eqref{eq:W-explicit}, each pair contributes $\|v\|^2$ and the integral gives $\tau^{2n}/(2n)!$, so $\cW=\sum_n(-1)^nC_n\|v\|^{2n}\tau^{2n}/(2n)!$. Finally, we note $J_1(2x)/x
=\sum_n(-1)^nC_n x^{2n}/(2n)!$.
\end{proof}

\begin{example}[$d=2$: ordered vs. symmetrised moments]\label{ex:d2}
For $d=2$ the ordered moments see information invisible to the commutative curve. For example, at degree $4$, the two non-crossing pairings of $\{1,2,3,4\}$ are $\pi_1=\{(1,2),(3,4)\}$ and $\pi_2=\{(1,4),(2,3)\}$, giving $M_{1122}=M_{1221}=1$ but $M_{1212}=0$. The degree-$4$ Wilson line is
\begin{equation}
  \cW(\gamma)\big|_{\mathrm{deg}\,4}
  = \int_{\Delta_4}\bigl[
  \langle\dd\gamma(t_1),\dd\gamma(t_2)\rangle\langle\dd\gamma(t_3),\dd\gamma(t_4)\rangle
  + \langle\dd\gamma(t_1),\dd\gamma(t_4)\rangle\langle\dd\gamma(t_2),\dd\gamma(t_3)\rangle\bigr];
\end{equation}
for a straight line the two terms in the integral coincide, but for a curved path they differ.
\end{example}

While \eqref{eq:W-explicit} can also be obtained by computing the planar moments directly \cite{nica2006lectures,cass2026free} and contracting against the signature, the spectral equation places the Wilson line inside a family of geometric objects---spectral curve, resolvent, topological recursion---which yields closed-form solutions, systematic $1/N$ corrections (as we see in the following Section \ref{sec:genus}), and a framework for interacting path developments (Section \ref{sec:discussion}).

\section{$1/N$ corrections}\label{sec:genus}

We now compute the $1/N$ corrections to the randomised path development, by studying corrections to the planar path development equation \eqref{eq:CT}. This planar large-$N$ limit equation arose from \eqref{eq:SD} -- the Schwinger--Dyson equation where the connected two-trace term was dropped. Keeping the subleading term in the same matrix model integration by parts (Section~\ref{sec:spectral}) gives the exact finite-$N$ identity (see, for example, \cite{migdal1983loop,eynard2016counting,guionnet2007second,guionnet2019asymptotics})
\begin{equation}\label{eq:SD-full}
  M^N_{I\,\mu} = -\sum_{I=K\mu L}\Bigl[M^N_K M^N_L   + \frac{1}{N^2}M^{\mathrm{conn}}_{K;\,L}\Bigr],
\end{equation}
where the first term is given by the finite-$N$ moment analogue of \eqref{eq:planar-moments}:
\begin{equation}
M^N_I = \frac{1}{N}\,i^{|I|}\langle\tr A_I\rangle
\end{equation}
and the second term is the connected two-trace correlator
\begin{equation}\label{eq:conn-def}
  M^{\mathrm{conn}}_{K;L} := i^{|K|+|L|}\Bigl[ \bigl\langle\tr(A_K)\,\tr(A_L)\bigr\rangle - \bigl\langle\tr(A_K)\bigr\rangle\bigl\langle\tr(A_L)\bigr\rangle\Bigr]\,.
\end{equation}

\subsection{The genus $g$ Wilson line}\label{sec:general-d}

Using the Schwinger-Dyson equations \eqref{eq:SD-full}, the finite $N$ Wilson line can then be expanded in inverse powers of $N$ as a sum over pair partitions on higher genus surfaces. 

The genus of a pair partition, which enters the closed-form sum below, is defined as follows (see for example \cite{anderson2010introduction, mingo2017free}).

\begin{definition}[Genus of a pair partition]\label{def:genus-pp}
A pair partition $\pi\in\Pair(2n)$ may be associated to a fixed-point-free involution $\sigma_\pi\in S_{2n}$. Writing the complete cycle as $\tau=(1\,2\,\cdots\,2n)$, the genus of $\pi$ is defined by
\begin{equation}\label{eq:genus-def}
    g(\pi) := \frac{n+1-c(\sigma_\pi\tau)}{2},
\end{equation}
where $c(\cdot)$ is the number of cycles. Equivalently $g(\pi)$ is the minimal genus of a surface with the chord diagram without crossings. The genus-0 pairings are non-crossing.
\end{definition}

\begin{theorem}[Genus-$g$ Wilson line]\label{thm:genus-g}
For $d\geq1$ and any smooth $\gamma$,
\begin{equation}\label{eq:genus-g}
  \cW_g(\gamma;s,t) = \sum_{n}(-1)^n \sum_{\substack{\pi\in\Pair(2n)\\g(\pi)=g}} \int_{\Delta_{2n}[s,t]} \prod_{(a,b)\in\pi}\langle\dd\gamma(t_a),\dd\gamma(t_b)\rangle,
\end{equation}
and for a matrix $U(\gamma;s,t)$ solution of \eqref{eq:dev} we have the following $1/N$ genus expansion for the path development:
\begin{equation}\label{eq:1overN-path-dev}
\cW(\gamma;s,t)=\frac{1}{N}\langle\tr U(\gamma)\rangle = \sum_{g\geq0}N^{-2g}\cW_g(\gamma)\,.
\end{equation}
\end{theorem}

\begin{proof}
For GUE matrices, Wick's theorem gives the moment exactly at finite $N$,
\begin{equation}
  \frac{1}{N}\bigl\langle\tr(A_{i_1}\cdots A_{i_{2n}})\bigr\rangle
  = \sum_{\pi\in\Pair(2n)} N^{-2g(\pi)}\prod_{(a,b)\in\pi}\delta_{i_a i_b},
\end{equation}
where the sum is over all colour-matched pair partitions, and $g(\pi)$ is the genus of Definition~\ref{def:genus-pp} \cite{t1993planar,eynard2016counting}. Collecting the $N^{-2g}$ coefficient gives the genus-$g$ moment 
\begin{equation}
M_{i_1\cdots i_{2n};g}=(-1)^n\sum_{\{\pi:g(\pi)=g\}}\prod\delta_{i_a i_b}\,,
\end{equation}
and the genus $g$ non-commutative directional resolvent (moment generating function) is obtained by summing over words:
\begin{equation}
    W_g=\sum_{I} M_{I;g}\xh_I
\end{equation}
Evaluating $\cW_g=\ev_\gamma(W_g)$ produces \eqref{eq:genus-g} and summing over all genus gives \eqref{eq:1overN-path-dev}. The $g=0$ term reproduces the free semicircular moments of Theorem~\ref{thm:nc2}.
\end{proof}

\subsection{The genus-\texorpdfstring{$g$}{g} loop equation on path space}\label{sec:genus-loop}

We now show that the pair-partition sum \eqref{eq:genus-g} is the closed-form combinatorial solution of a recursion on path space that generalises the planar path development equation of theorem \ref{thm:CT}.

Start from the finite-$N$ generating function
\begin{equation}
W=\sum_I M^N_I\,\xh_I\,.
\end{equation}
Contracting the finite $N$ Schwinger--Dyson identity \eqref{eq:SD-full} with $\xh_{i_1}\cdots\xh_{i_n}$ and reassembling the blocks exactly as in the proof of Theorem~\ref{thm:nc-spectral}---now retaining the connected term---gives 
\begin{equation}\label{eq:nc-spectral-finN}
  W = 1 - W\cdot\Phi(W) - \frac{1}{N^{2}}\,\Delta\,, \qquad \Delta = \sum_{\mu}\sum_{K,L} M^{\mathrm{conn}}_{K;\,L}\;
  \xh_K\,\xh_\mu\,\xh_L\,\xh_\mu.
\end{equation}
Expanding \eqref{eq:nc-spectral-finN} in genus, with
\begin{equation}\label{eq:genus-exp-W-Delta}
W=\sum_g N^{-2g}W_g\,, \quad \Delta=\sum_g N^{-2g}\Delta_g
\end{equation}
where $\Delta_g$ is the genus-$g$ part of the two-trace source, gives a hierarchy of equations for the corrections $\cW_g=\ev_\gamma(W_g)$ in \eqref{eq:1overN-path-dev}, which we now show are linear on path space.

\begin{proposition}[Genus-$g$ loop equation]\label{thm:genus-g-loop}
For each $g\geq1$, the genus-$g$ correction $\cW_g$ in \eqref{eq:genus-g} satisfies the \emph{linear} Volterra equation
\begin{equation}\label{eq:genus-g-loop}
  \cW_g = L_0\,\cW_g + \sigma_g,
  \qquad
  \sigma_g = -\sum_{\substack{g_1+g_2=g\\g_1,g_2\geq1}}\cW_{g_1}\star_\gamma\cW_{g_2} - \delta_g,
\end{equation}
where 
\begin{equation}\label{eq:lin-convolution}
L_0 f := -(\cW_0\star_\gamma f + f\star_\gamma\cW_0)
\end{equation}
is the linearisation of the path convolution about the planar solution and $\delta_g=\ev_\gamma(\Delta_{g-1})$.
\end{proposition}

\begin{proof}
Collect $O(N^{-2g})$ in \eqref{eq:nc-spectral-finN} using the expansions \eqref{eq:genus-exp-W-Delta}. Since the operator $Phi$ in \eqref{eq:def-Phi} is linear,
\begin{equation}
W\cdot\Phi(W)=\sum_g N^{-2g}\sum_{g_1+g_2=g}W_{g_1}\cdot\Phi(W_{g_2})\,,
\end{equation}
and the $O(N^{-2g})$ of \eqref{eq:nc-spectral-finN} becomes $W_g=-\sum_{g_1+g_2=g}W_{g_1}\cdot\Phi(W_{g_2})-\Delta_{g-1}$. The only terms containing $W_g$ on the r.h.s. are $g_1=0$ and $g_2=0$, each linear in it, so 
\begin{equation}
W_g = -W_0\cdot\Phi(W_g) - W_g\cdot\Phi(W_0)+\Sigma_g\,,\quad \Sigma_g=-\sum_{g_1,g_2\geq1}W_{g_1}\cdot\Phi(W_{g_2})-\Delta_{g-1}\,.
\end{equation}
From the definition of $W$ \eqref{eq:R-nc} and $\Phi$ \eqref{eq:def-Phi}:
\begin{equation}
   W_{g_1}\cdot\Phi(W_{g_2})= \cL_B(\cW_{g_1})\cdot\Phi\bigl(\cL_B(\cW_{g_2})\bigr)
  = \sum_{\mu=1}^d \cL_\mu(\cW_{g_1})\,\cL_\mu(\cW_{g_2})\,.
\end{equation}
and now applying $\ev_\gamma$ and using Theorem~\ref{prop:ev-props} by linearity we have
\begin{equation}
\ev_\gamma\left(W_0\cdot\Phi(W_g)+W_g\cdot\Phi(W_0)\right)=-L_0\cW_g\,,\quad \ev_\gamma(\Sigma_g)=\sigma_g\,,
\end{equation}
finally giving \eqref{eq:genus-g-loop}.
\end{proof}
The sources of equations \eqref{eq:genus-g-loop} remain to be determined, by going deeper in the Schwinger--Dyson hierarchy. We now pursue this for genus one. 

\subsubsection{The genus one system}\label{ex:g01}

We focus on the genus one correction in particular. At $g=0$, Theorem \ref{thm:genus-g-loop} is the planar equation \eqref{eq:CT}:
\begin{equation}\label{eq:W0-int}
  \cW_0(s,t) = 1 - \int_s^t\!\!\int_s^r \cW_0(s,u)\,\cW_0(u,r)\,\langle\mathrm d\gamma(u),\mathrm d\gamma(r)\rangle,
  \qquad \cW_0(s,s)=1.
\end{equation}
For $g\ge 1$ we introduce new notation for the insertion of multiple traces of developments into the Wilson loop, corresponding to the  connected $n$-point functions of the corresponding sub-developments when the path is split into $n$ parts. We write $\cW_{g,n}$ for the genus-$g$, $n$-insertion correlator, by analogy with the typical topological recursion notation. We abbreviate $\cW_g:=\cW_{g,1}$ for the single insertion Wilson line. In particular, the planar connected two-point function of two sub-developments is given explicitly by
\begin{equation}\label{eq:W2-def}
  \cW_{0,2}(a,b;c,d) = \lim_{N\to\infty} N^2\Bigl(\bigl\langle\tfrac1N\tr U_{a,b}\,\tfrac1N\tr U_{c,d}\bigr\rangle - \bigl\langle\tfrac1N\tr U_{a,b}\bigr\rangle\bigl\langle\tfrac1N\tr U_{c,d}\bigr\rangle\Bigr),
\end{equation}
with $U_{a,b}$ being the development of the restriction $\gamma(t)|_{t\in[a,b]}$.

At $g=1$ the equation \eqref{eq:genus-g-loop} is linear in $\cW_1$,
\begin{equation}\label{eq:W1-int}
\begin{split}
  \cW_1(s,t) =& - \int_s^t \int_s^r \Bigl(\cW_1(s,u)\,\cW_0(u,r) + \cW_0(s,u)\,\cW_1(u,r) + \cW_{0,2}(s,u;u,r)\Bigr)\,\langle\mathrm d\gamma(u),\mathrm d\gamma(r)\rangle\,;\\
 \cW_1(s,s)=& 0,
\end{split}
\end{equation}
where the first two terms on the right are $L_0\cW_1$, the linearisation \eqref{eq:lin-convolution} of the quadratic term of \eqref{eq:W0-int} about the planar solution, and the source $\cW_{0,2}$ is the planar connected two-point function of two sub-developments given in \eqref{eq:W2-def}. This is not an independent unknown, in fact it satisfies its own linear loop equation, driven by $\cW_0$. Applying the two-trace Schwinger--Dyson identity (see \ref{lem:two-trace}) and evaluating against the signature (details in appendix \ref{app:closed}) gives
\begin{equation}\label{eq:W2-loop-main}
\begin{split}
  \cW_{0,2}(a,b;c,d) = {}& - \int_a^b \int_a^r \Bigl(\cW_0(a,u)\,\cW_{0,2}(u,r;c,d) + \cW_{0,2}(a,u;c,d)\,\cW_0(u,r)\Bigr)\,\langle\mathrm d\gamma(u),\mathrm d\gamma(r)\rangle \\
  & - \int_a^b \int_c^d \cW_0\bigl(\gamma|_{[a,u]}*\gamma|_{[v,d]}*\gamma|_{[c,v]}\bigr)\,\langle\mathrm d\gamma(u),\mathrm d\gamma(v)\rangle.
\end{split}
\end{equation}
Here we have used the notation $\gamma|_{[a,b]}$ to denote $\gamma(t)|_{t\in[a,b]}$. We have obtained a closed system of integral equations at genus one: \eqref{eq:W0-int} fixes $\cW_0$; \eqref{eq:W2-loop-main}, sourced by $\cW_0$, fixes $\cW_{0,2}$; and \eqref{eq:W1-int}, sourced by $\cW_{0,2}$, in turn determines $\cW_1$. We solve these equations using boundary data $\cW_0(a,a)=1$ and $\cW_{0,2}(a,a;c,d)=\cW_1(a,a)=0$.

In the following section \ref{sec:numerics} we solve the system by numerical discretisation for a number of example paths finding good agreement with the Monte Carlo method and with the closed-form combinatorial formula of theorem \ref{thm:genus-g}.

\begin{remark}\label{rem:W2-closed}
The machinery of topological recursion closes the hierarchy at higher genus more generally. The genus-$g$ source $\delta_g$ in \eqref{eq:genus-g-loop} is built from connected multi-trace correlators, each of which satisfies its own loop equation, obtained by applying the Schwinger--Dyson identity to a product of several traces and evaluating against the signature. For $d=1$ this is the standard Eynard--Orantin topological recursion \cite{eynard2007invariants,eynard2016counting} on the spectral curve \eqref{eq:semicircle-res-eq}, generating all correlators from the Bergman kernel. For $d\geq2$ the correlators are ordered words, and closing the system for all $g$ requires a non-commutative lift, producing these ordered correlators directly in the tensor algebra. We do not pursue the general approach here, focussing on the genus one corrections which already capture the main ideas of topological recursion.
\end{remark}

\subsection{Numerical validation}\label{sec:numerics}

We now verify proposition \ref{thm:genus-g-loop} and the system of loop equations that follow in the genus one case by three independent methods. We solve the discretised loop equations, evaluate the combinatorial closed form of Theorem~\ref{thm:genus-g} directly, and Monte Carlo sample the finite-$N$ development. As examples we take three paths:
\begin{itemize}
\item a $d=2$ quarter-circle $\gamma(t)=s(\cos\tfrac{\pi t}{2},\sin\tfrac{\pi t}{2})$, $s=1.0$;
\item a $d=3$ twisted cubic $\gamma(t)=a(t,t^2,t^3)$, $a=0.75$;
\item a $d=2$ closed circle $\gamma(t)=r(\cos 2\pi t,\sin 2\pi t)$, $r=0.5$: a Wilson loop.
\end{itemize}

\paragraph{Discrete loop equations.}
Let $\sigma=(\sigma_1,\dots,\sigma_r)$ be an ordered sequence of grid
steps, with increments $\Delta\gamma_k=\gamma(t_k)-\gamma(t_{k-1})$. We write
\begin{equation}
\sigma'=(\sigma_1,\dots,\sigma_{r-1}),\qquad
\sigma_{<\ell}=(\sigma_1,\dots,\sigma_{\ell-1}),\qquad
\sigma_{>\ell}=(\sigma_{\ell+1},\dots,\sigma_{r-1}),
\end{equation}
The explicit Euler discretisation of the closed genus one system is
\begin{align}
  K_0[\sigma]&= K_0[\sigma'] - \sum_{\ell=1}^{r-1}\bigl\langle\Delta\gamma_{\sigma_\ell}, \Delta\gamma_{\sigma_r}\bigr\rangle K_0[\sigma_{<\ell}]K_0[\sigma_{>\ell}],
  \label{eq:euler-ct}\\
  c_2[\sigma;\rho] &= c_2[\sigma';\rho] - \sum_{\ell=1}^{r-1} \bigl\langle\Delta\gamma_{\sigma_\ell}, \Delta\gamma_{\sigma_r}\bigr\rangle \Bigl(K_0[\sigma_{<\ell}]c_2[\sigma_{>\ell};\rho]+c_2[\sigma_{<\ell};\rho]K_0[\sigma_{>\ell}]\Bigr)\\
  &\qquad - \sum_{m=1}^{|\rho|} \bigl\langle\Delta\gamma_{\rho_m}, \Delta\gamma_{\sigma_r}\bigr\rangle K_0[\sigma'*\rho_{\mathrm{rot}(m)}],
  \label{eq:c2-disc}\\
  K_1[\sigma] &= K_1[\sigma'] - \sum_{\ell=1}^{r-1} \bigl\langle\Delta\gamma_{\sigma_\ell},\Delta\gamma_{\sigma_r}\bigr\rangle \Bigl(
    K_1[\sigma_{<\ell}]K_0[\sigma_{>\ell}] +K_0[\sigma_{<\ell}]K_1[\sigma_{>\ell}] +c_2[\sigma_{<\ell};\sigma_{>\ell}]\Bigr).
  \label{eq:K1-disc}
\end{align}
where $K_0[\varnothing]=1$, $c_2[\varnothing;\rho]=0$ and $K_1[\varnothing]=0$, and
\begin{equation}
\rho_{\mathrm{rot}(m)} =(\rho_{m+1},\dots,\rho_{|\rho|}, \rho_1,\dots,\rho_{m-1}).
\end{equation}
The restriction $\ell<r$ is the discrete counterpart of the  time ordering $u<r$ in the integrals. The inner product pairs
the increments assigned to the $\ell$-th and final grid cells. The last term of \eqref{eq:c2-disc} discretises the source \eqref{eq:beta0} contracting the final step of the first boundary with the $m$-th step of the second boundary into the single concatenated path represented by $\sigma'*\rho_{\mathrm{rot}(m)}$.

Equations \eqref{eq:euler-ct}, \eqref{eq:c2-disc}, and
\eqref{eq:K1-disc} are solved successively by forward recurrence and we thereby obtain numerical approximations to $\mathcal{W}_0$ and $\mathcal{W}_1$ (via $\mathcal{W}_{0,2}$).

\paragraph{Combinatorial closed form.}
We also evaluate the pair-partition solution of Theorem \ref{thm:genus-g} directly, computing the exact signature of the piecewise-linear path by Chen's identity (see section \ref{sec:path-sig}) and contracting it against the colour-matched pairings of each genus,
\begin{equation}\label{eq:W1-comb}
  \cW_g = \sum_{n}(-1)^n \sum_{\substack{\pi\in\Pair(2n)\\ g(\pi)=g}}\ \sum_{i_1,\dots,i_{2n}} S^{i_1\cdots i_{2n}}(\gamma)\prod_{(a,b)\in\pi}\delta_{i_a i_b},
\end{equation}
truncated in the examples such that the last retained terms are of order $10^{-4}$ or below. This is the most expensive method numerically due to the combinatorial increase in pair partitions.

\paragraph{Monte Carlo.}
Sampling $d$ independent $N\times N$ GUE matrices with entries of variance $1/N$ (the law of Section~\ref{sec:setup}), we form the exact development of the piecewise-linear path,
\begin{equation}\label{eq:euler-prod}
  U = \prod_{k=0}^{n_{\mathrm{grid}}-1}\exp \bigl(iA_\mu\Delta\gamma^\mu_k\bigr),
\end{equation}
at $n_{\mathrm{grid}}=48$, and average $\tfrac{1}{N}\mathrm{Re}\,\tr U$ over $0.6$--$2.4\times10^5$ realisations for each $N=2,\ldots,8$.

\paragraph{Results.}
We verify the genus one system in two ways. Firstly, we
compute $\cW_0$ and $\cW_1$ both from the discretised loop equations and from the combinatorial pair-partition formula of theorem \ref{thm:genus-g}. Secondlt, as an independent
we compare the truncated expansion $\cW_0+N^{-2}\cW_1$ with Monte Carlo estimates of the finite $N$ matrix development. We find good agreement on the example three paths:
\begin{center}
\begin{tabular}{lccc}
 & $\cW_0$ & $\cW_1$ (loop equations) & $\cW_1$ (combinatorial)\\
\hline
quarter-circle & $0.29200$ & $+0.071$ & $+0.07163$\\
twisted cubic  & $0.36738$ & $+0.060$ & $+0.06024$\\
circle         & $0.66937$ & $+0.332$ & $+0.33196$\\
\end{tabular}
\end{center}
Figure \ref{fig:volterra} further compares $\cW_0+\cW_1/N^2$ with the Monte Carlo Wilson line across $N=2$--$8$: the genus-$1$ term accounts for the deviation from the planar value on all three paths. We thus verify theorem \ref{thm:genus-g-loop} and the closure of the hierarchy at genus one.

\begin{figure}[ht]
\centering
\includegraphics[width=\textwidth]{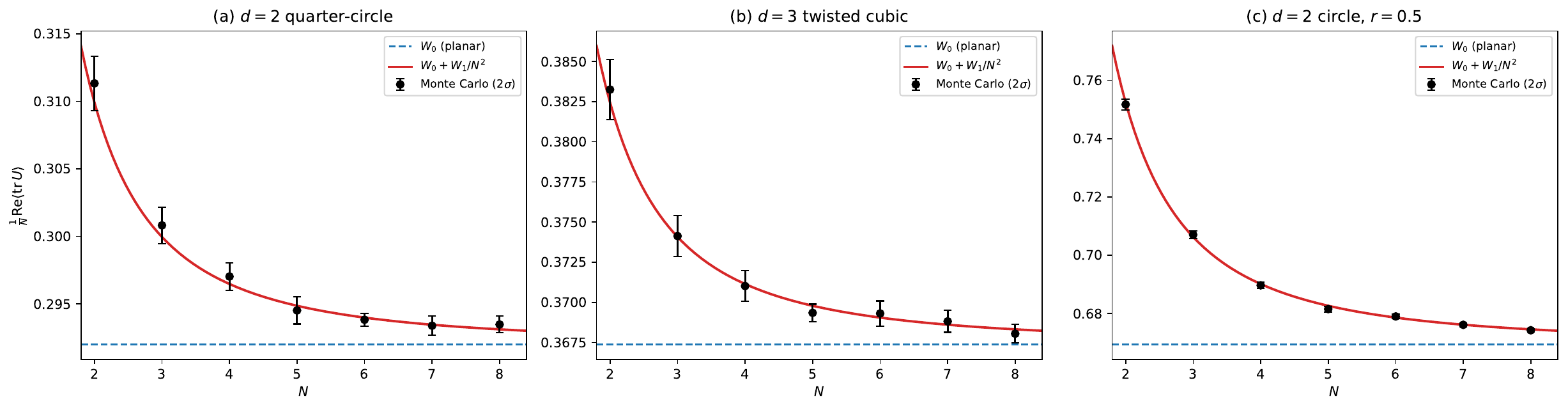}
\caption{Genus-$1$ correction to the finite-$N$ Wilson line ($2\sigma$ error bars). Monte Carlo values (points); $\cW_0+\cW_1/N^2$ (solid) and the planar value $\cW_0$ (dashed), for (a) the $d=2$ quarter-circle ($s=1.0$, $\cW_1=+0.0716$), (b) the $d=3$ twisted cubic ($a=0.75$, $\cW_1=+0.0602$) and (c) the $d=2$ Wilson loop ($r=0.5$, $\cW_1=+0.3320$).}
\label{fig:volterra}
\end{figure}

\section{Discussion}\label{sec:discussion}

In this work we have introduced a non-commutative Laplace transform between path functionals and the tensor algebra, together with a new convolution product on path space. The Laplace transform maps the path convolution to the concatenation product in the tensor algebra (analogously, the same transform maps the pointwise product of functionals to the shuffle product; Proposition~\ref{prop:ev-props}). As an application, performing the transform on randomised path developments relates the path-dependent Wilson line $\cW(\gamma)$ of Cass and Turner \cite{cass2026free} to an algebraic fixed-point equation in the tensor algebra. Together with topological recursion, the planar path development equation, its genus-$g$ extensions, and the closed-form combinatorial solutions are then all direct consequences of this algebraic structure. We conclude with a discussion of practical implications and directions for future work.

\paragraph{Kernel methods.}
From the perspective of kernel methods this framework offers two practical benefits. First, the spectral curve gives closed-form access to the Schwinger--Dyson kernel of \cite{cass2026free}, and the $1/N$ corrections quantify finite-$N$ approximation errors relevant to implementations \cite{lou2023pcf}. In practice, randomised path development features are computed at finite $N$, and the genus expansion $\cW(\gamma) = \cW_0(\gamma) + N^{-2}\cW_1(\gamma) + \cdots$ provides an asymptotic series. This suggests that (i) \textit{choosing $N$} for a desired kernel accuracy $\epsilon$, the genus-$1$ correction determines the minimal dimension via $N\gtrsim|\cW_1/\epsilon|^{1/2}$; (ii) \emph{enriched features}: $\cW_1(\gamma)$ is itself a path functional, probing the path through crossing pair partitions, $(\cW_0,\cW_1)$ forms a richer feature vector than $\cW_0$ alone; and (iii) \emph{quantum implementations}: the quantum signature kernel of \cite{crew2025quantum} is computed at finite $N$ on physical hardware, where $N$ is set by the logical qubit count, and the genus expansion quantifies the resulting finite-size effects. 

\paragraph{Interacting models.}\label{sec:interact}
The framework is also expected to extend beyond Gaussian random matrices. For a potential $V=\frac{1}{2}\sum_\mu A_\mu^2+V_{\mathrm{int}}$, integration by parts produces, alongside the universal Gaussian splitting $W\cdot\Phi(W)$, a correction encoding $\partial V_{\mathrm{int}}/\partial A_\mu$, deforming the spectral equation \eqref{eq:nc-spectral} to
\begin{equation}\label{eq:nc-int}
  W = 1 - W\cdot\Phi(W) - \sum_\mu G_\mu^{(\mathrm{int})}\xh_\mu,
\end{equation}
where $G_\mu^{(\mathrm{int})}$ generates the interaction-induced correlators. Applying $\ev_\gamma$ gives a generalised path development equation,
\begin{equation}\label{eq:gen-CT}
  \cW(\gamma;s,t) = 1 - \cW\star_\gamma\cW - \cW_{\mathrm{int}}(\gamma;s,t), \qquad \cW_{\mathrm{int}} = \sum_\mu\ev_\gamma\bigl(G_\mu^{(\mathrm{int})}\xh_\mu\bigr),
\end{equation}
with $V_{\mathrm{int}}=0$ recovering \eqref{eq:CT} in the large-$N$ limit. Unlike the Gaussian case, \eqref{eq:gen-CT} is not closed: $\cW_{\mathrm{int}}$ is determined by the full Schwinger--Dyson hierarchy. We have not developed its solution theory here but in principle perturbation theory may be applied.

\paragraph{Resurgence.}
The $1/N$ expansion is asymptotic, and its non-perturbative completion--- instanton contributions from eigenvalue tunnelling on the spectral curve--- should have a path-space counterpart. This is made suggestive by the tensor-algebra formulation, since the signature and \'{E}calle's resurgence monomials are both symmetral characters of different shuffle Hopf algebras \cite{fauvet2017ecalle}; it would be interesting to investigate the interplay of these two structures more closely---mould calculus and arborification on the one hand, and the branched rough path story on the other \cite{curry2020planarly}.

\subsection*{Acknowledgements}

We thank Masazumi Honda, Cris Salvi, Kurusch Ebrahimi-Fard and Will Turner for interesting discussions. SC additionally thanks Andrea Di Fini (CESQ) and colleagues Asier Pi\~neiro Orioli, Shannon Whitlock, Tom Hartweg, Hugo Perrin and Adrian Aasen at QPerfect for many stimulating conversations. IA would like to thank Heng-Yu Chen for the invitation to visit National Taiwan University, where the collaboration with SC on this project began. This work was conducted in part during an Oberwolfach Research in Pairs programme. The code used in the numerical validation section was written with the assistance of Anthropic's Claude (Opus 4.5). IA was partially supported by UKRI EPSRC Early Career Fellowship EP/S004076/1. TC is supported in part by UK Research and Innovation (UKRI) through the Engineering and Physical Sciences Research Council (EPSRC) via Programme Grants [Grant No. UKRI1010: High order mathematical and computational infrastructure for streamed data that enhance contemporary generative and large language models].

\appendix

\section{The closed genus-one system}
\label{app:closed}

In this appendix we derive the loop equation for the planar two-point function $\cW_{0,2}$ of \eqref{eq:W2-def}, showing that it is determined by the planar Wilson line $\cW_0$ alone. Together with the genus-one equation \eqref{eq:W1-int} this closes the loop equation hierarchy at genus one. The result is a consequence of the two trace analogue of the Schwinger--Dyson equation.

\subsection{The two-trace Schwinger--Dyson identity}

As in the main text, $A_1,\dots,A_d$ are independent GUE matrices with law $\propto\exp(-\tfrac N2\sum_\nu\tr A_\nu^2)$, $\langle\,\cdot\,\rangle$ is the corresponding average, and $A_I=A_{i_1}\cdots A_{i_n}$ for a word $I=\{i_1,\cdots,i_n\}$. Gaussian integration by parts now acts on a product of two traces.

\begin{lemma}\label{lem:two-trace}
For any words $I,J$ and any colour $\mu$, at every finite $N$,
\begin{equation}\label{eq:two-trace}
  \bigl\langle\tr(A_I A_\mu)\,\tr A_J\bigr\rangle = \frac1N\sum_{I=K\mu L}\bigl\langle\tr A_K\,\tr A_L\,\tr A_J\bigr\rangle + \frac1N\sum_{J=P\mu Q}\bigl\langle\tr(A_I\,A_Q A_P)\bigr\rangle.
\end{equation}
\end{lemma}
\begin{proof}
This is the standard two-trace Schwinger--Dyson equation \cite{migdal1983loop,ambjorn1990multiloop,eynard2016counting}.
\end{proof}

\subsection{The planar two-point function}

Taking the connected part of \eqref{eq:two-trace} and then its planar limit gives a linear equation for the planar two-point function. 

\begin{proposition}[Loop equation for the two-point function]
\label{prop:W2-loop}
The planar two-point function $\cW_{0,2}(a,b;c,d)$ satisfies
\begin{equation}\label{eq:W2-loop}
  \cW_{0,2}=L_0\cW_{0,2}+\beta_0,
\end{equation}
where
\begin{equation}\label{eq:beta0}
\beta_0(a,b;c,d) = -\int_a^b\!\!\int_c^d \bigl\langle\dd\gamma(u),\dd\gamma(v)\bigr\rangle \cW_0\bigl( \gamma|_{[a,u]}* \gamma|_{[v,d]}* \gamma|_{[c,v]} \bigr).
\end{equation}
Hence $\cW_{0,2}$ is determined by $\cW_0$.
\end{proposition}

\begin{proof}
The connected correlator is
\begin{equation}
\langle X;Y\rangle_{\mathrm c}
=
\langle XY\rangle-\langle X\rangle\langle Y\rangle
\end{equation}
Subtracting the product of the one-trace identity with $\langle\tr A_J\rangle$ from \eqref{eq:two-trace} gives
\begin{equation}\label{eq:two-trace-conn}
\begin{aligned}
\bigl\langle\tr(A_I A_\mu);\tr A_J\bigr\rangle_{\mathrm c} =& \frac{1}{N}\sum_{I=K\mu L} \bigl\langle\tr A_K\,\tr A_L;\tr A_J\bigr\rangle_{\mathrm c}
\\ &+ \frac{1}{N}\sum_{J=P\mu Q} \bigl\langle\tr(A_I A_Q A_P)\bigr\rangle.
\end{aligned}
\end{equation}
At leading order in $N$,
\begin{equation}\label{eq:three-trace-factorisation}
\begin{aligned}
\bigl\langle\tr A_K\,\tr A_L;\tr A_J\bigr\rangle_{\mathrm c}
=& \langle\tr A_K\rangle \bigl\langle\tr A_L;\tr A_J\bigr\rangle_{\mathrm c}
\\ &+ \bigl\langle\tr A_K;\tr A_J\bigr\rangle_{\mathrm c} \langle\tr A_L\rangle +O(N^{-1}).
\end{aligned}
\end{equation}
Define the planar connected moments by
\begin{equation}\label{eq:planar-connected-moments}
\begin{aligned}
M^{(0),\mathrm c}_{I;J}
&=
\lim_{N\to\infty}
\bigl\langle
\tr\bigl((iA)_I\bigr);
\tr\bigl((iA)_J\bigr)
\bigr\rangle_{\mathrm c}
\end{aligned}
\end{equation}
Using \eqref{eq:three-trace-factorisation} in \eqref{eq:two-trace-conn} gives
\begin{equation}\label{eq:two-trace-planar-recursion}
M^{(0),\mathrm c}_{I\mu;J} = -\sum_{I=K\mu L} \left( M_K M^{(0),\mathrm c}_{L;J} + M^{(0),\mathrm c}_{K;J}M_L \right) -\sum_{J=P\mu Q}M_{IQP}.
\end{equation}
The signs come from the factor $i^2=-1$ associated with each
contraction.

Expanding the two developments against their signatures gives
\begin{equation}\label{eq:W02-signature-expansion}
\cW_{0,2}(a,b;c,d) = \sum_{I,J} M^{(0),\mathrm c}_{I;J} S^I_{a,b}S^J_{c,d}.
\end{equation}
Multiply \eqref{eq:two-trace-planar-recursion} by $S^{I\mu}_{a,b}S^J_{c,d}$ and sum over $I,J,\mu$. The first sum gives
\begin{equation}
-\int_{a<u<r<b} \bigl\langle\dd\gamma(u),\dd\gamma(r)\bigr\rangle \left[ \cW_0(a,u)\cW_{0,2}(u,r;c,d) + \cW_{0,2}(a,u;c,d)\cW_0(u,r) \right],
\end{equation}
which is $L_0\cW_{0,2}$.

In the second sum, the occurrence of $\mu$ in $J=P\mu Q$ marks a point $v$ on the second boundary. Summing over $\mu$ and using Chen's identity gives
\begin{equation}
-\int_a^b\!\!\int_c^d \bigl\langle\dd\gamma(u),\dd\gamma(v)\bigr\rangle \cW_0\bigl( \gamma|_{[a,u]}* \gamma|_{[v,d]}*\gamma|_{[c,v]} \bigr),
\end{equation}
which is $\beta_0$. 
\end{proof}

\subsection{The closed hierarchy}

Collecting Proposition \ref{prop:W2-loop} with Section \ref{ex:g01} and writing the integrals explicitly we see that the genus one corrections solve the loop equation system
\begin{align}
  \cW_0(a,b) &= 1 - \int_{a<u<r<b}\bigl\langle\dd\gamma(u),\dd\gamma(r)\bigr\rangle\, \cW_0(a,u)\,\cW_0(u,r), \label{eq:int-0-app}\\
    \cW_{0,2}(a,b;c,d) &= -\int_{a<u<r<b}\bigl\langle\dd\gamma(u),\dd\gamma(r)\bigr\rangle\, \Bigl[\cW_0(a,u)\,\cW_{0,2}(u,r;c,d) + \cW_{0,2}(a,u;c,d)\,\cW_0(u,r)\Bigr] \notag\\
  &\qquad -\int_a^b\!\!\int_c^d\bigl\langle\dd\gamma(u),\dd\gamma(v)\bigr\rangle\,\cW_0\bigl(\gamma|_{[a,u]}*\gamma|_{[v,d]}*\gamma|_{[c,v]}\bigr), \label{eq:int-2}\\
  \cW_1(a,b) &= -\int_{a<u<r<b}\bigl\langle\dd\gamma(u),\dd\gamma(r)\bigr\rangle\, \Bigl[\cW_0(a,u)\,\cW_1(u,r) + \cW_1(a,u)\,\cW_0(u,r)\Bigr]\nonumber\\
  & \hspace{15pt}-\int_{a<u<r<b}\bigl\langle\dd\gamma(u),\dd\gamma(r)\bigr\rangle\, \cW_{0,2}(a,u;u,r)\label{eq:int-1-app}
\end{align}
with boundary data $\cW_0(a,a)=1$ and $\cW_{0,2}(a,a;c,d)=\cW_1(a,a)=0$. The equations are solved in order: the first equation determines $\cW_0$, then the second determines $\cW_{0,2}$, and finally we sole the third equation to determine $\cW_1$.

\bibliographystyle{plain}
\bibliography{sig_laplace}

\begin{thebibliography}{10}

\bibitem{ambjorn1990multiloop}
Jan Ambj{\o}rn, Jerzy Jurkiewicz, and Yu~M Makeenko.
\newblock Multiloop correlators for two-dimensional quantum gravity.
\newblock {\em Physics Letters B}, 251(4):517--524, 1990.

\bibitem{anderson2010introduction}
Greg~W Anderson, Alice Guionnet, and Ofer Zeitouni.
\newblock {\em An introduction to random matrices}.
\newblock Number 118. Cambridge university press, 2010.

\bibitem{cass2024lecture}
Thomas Cass and Cristopher Salvi.
\newblock Lecture notes on rough paths and applications to machine learning.
\newblock {\em arXiv preprint arXiv:2404.06583}, 2024.

\bibitem{cass2026free}
Thomas Cass and William~F Turner.
\newblock Free probability, path developments and signature kernels as
  universal scaling limits.
\newblock {\em The Annals of Applied Probability}, 36(2):1082--1109, 2026.

\bibitem{chen1954iterated}
Kuo-Tsai Chen.
\newblock Iterated integrals and exponential homomorphisms.
\newblock {\em Proceedings of the London Mathematical Society}, 3(1):502--512,
  1954.

\bibitem{chen1957integration}
Kuo-Tsai Chen.
\newblock Integration of paths, geometric invariants and a generalized
  baker-hausdorff formula.
\newblock {\em Annals of Mathematics}, 65(1):163--178, 1957.

\bibitem{chevyrev2025primer}
Ilya Chevyrev and Andrey Kormilitzin.
\newblock A primer on the signature method in machine learning.
\newblock In {\em Signature Methods in Finance: An Introduction with
  Computational Applications}, pages 3--64. Springer, 2025.

\bibitem{chevyrev2016characteristic}
Ilya Chevyrev and Terry Lyons.
\newblock Characteristic functions of measures on geometric rough paths.
\newblock 2016.

\bibitem{cirone2023neural}
Nicola~Muca Cirone, Maud Lemercier, and Cristopher Salvi.
\newblock Neural signature kernels as infinite-width-depth-limits of controlled
  resnets.
\newblock In {\em International Conference on Machine Learning}, pages
  25358--25425. PMLR, 2023.

\bibitem{crew2025quantum}
Samuel Crew, Cristopher Salvi, William~F Turner, Thomas Cass, and Antoine
  Jacquier.
\newblock Quantum path signatures.
\newblock {\em arXiv preprint arXiv:2508.05103}, 2025.

\bibitem{curry2020planarly}
Charles Curry, Kurusch Ebrahimi-Fard, Dominique Manchon, and Hans~Z
  Munthe-Kaas.
\newblock Planarly branched rough paths and rough differential equations on
  homogeneous spaces.
\newblock {\em Journal of Differential Equations}, 269(11):9740--9782, 2020.

\bibitem{eynard2016counting}
Bertrand Eynard et~al.
\newblock Counting surfaces.
\newblock {\em Progress in Mathematical Physics}, 70:414, 2016.

\bibitem{eynard2007invariants}
Bertrand Eynard and Nicolas Orantin.
\newblock Invariants of algebraic curves and topological expansion.
\newblock {\em Communications in Number Theory and Physics}, 1(2):347--452,
  2007.

\bibitem{fauvet2017ecalle}
Fr{\'e}d{\'e}ric Fauvet and Fr{\'e}d{\'e}ric Menous.
\newblock Ecalle's arborification-coarborification transforms and
  connes-kreimer hopf algebra.
\newblock In {\em Annales scientifiques de l'{\'E}cole Normale Sup{\'e}rieure},
  volume~50, pages 39--83, 2017.

\bibitem{guionnet2019asymptotics}
Alice Guionnet.
\newblock {\em Asymptotics of random matrices and related models: the uses of
  Dyson-Schwinger equations}, volume 130.
\newblock American Mathematical Soc., 2019.

\bibitem{guionnet2007second}
Alice Guionnet and Edouard Maurel-Segala.
\newblock Second order asymptotics for matrix models.
\newblock 2007.

\bibitem{hambly2010uniqueness}
Ben Hambly and Terry Lyons.
\newblock Uniqueness for the signature of a path of bounded variation and the
  reduced path group.
\newblock {\em Annals of Mathematics}, pages 109--167, 2010.

\bibitem{kidger2019deep}
Patrick Kidger, Patric Bonnier, Imanol Perez~Arribas, Cristopher Salvi, and
  Terry Lyons.
\newblock Deep signature transforms.
\newblock {\em Advances in neural information processing systems}, 32, 2019.

\bibitem{kiraly2019kernels}
Franz~J Kir{\'a}ly and Harald Oberhauser.
\newblock Kernels for sequentially ordered data.
\newblock {\em Journal of Machine Learning Research}, 20(31):1--45, 2019.

\bibitem{lou2023pcf}
Hang Lou, Siran Li, and Hao Ni.
\newblock Pcf-gan: generating sequential data via the characteristic function
  of measures on the path space.
\newblock {\em Advances in Neural Information Processing Systems},
  36:39755--39781, 2023.

\bibitem{lyons1998differential}
Terry~J Lyons.
\newblock Differential equations driven by rough signals.
\newblock {\em Revista Matem{\'a}tica Iberoamericana}, 14(2):215--310, 1998.

\bibitem{lyons2007differential}
Terry~J Lyons, Michael Caruana, and Thierry L{\'e}vy.
\newblock {\em Differential equations driven by rough paths: Ecole d'Et{\'e} de
  Probabilit{\'e}s de Saint-Flour XXXIV-2004}.
\newblock Springer, 2007.

\bibitem{migdal1983loop}
Aleksandr~Arkad'evi{\v{c}} Migdal.
\newblock Loop equations and 1n expansion.
\newblock {\em Physics Reports}, 102(4):199--290, 1983.

\bibitem{mingo2017free}
James~A Mingo and Roland Speicher.
\newblock {\em Free probability and random matrices}, volume~35.
\newblock Springer, 2017.

\bibitem{nica2006lectures}
Alexandru Nica and Roland Speicher.
\newblock {\em Lectures on the combinatorics of free probability}, volume~13.
\newblock Cambridge University Press, 2006.

\bibitem{reutenauer2003free}
Christophe Reutenauer.
\newblock Free lie algebras.
\newblock In {\em Handbook of algebra}, volume~3, pages 887--903. Elsevier,
  2003.

\bibitem{salvi2021signature}
Cristopher Salvi, Thomas Cass, James Foster, Terry Lyons, and Weixin Yang.
\newblock The signature kernel is the solution of a goursat pde.
\newblock {\em SIAM Journal on Mathematics of Data Science}, 3(3):873--899,
  2021.

\bibitem{t1993planar}
Gerard 't~Hooft.
\newblock A planar diagram theory for strong interactions.
\newblock In {\em The Large N Expansion In Quantum Field Theory And Statistical
  Physics: From Spin Systems to 2-Dimensional Gravity}, pages 80--92. World
  Scientific, 1993.

\bibitem{voiculescu1986addition}
Dan Voiculescu.
\newblock Addition of certain non-commuting random variables.
\newblock {\em Journal of functional analysis}, 66(3):323--346, 1986.

\bibitem{wigner1993characteristic}
Eugene~P Wigner.
\newblock Characteristic vectors of bordered matrices with infinite dimensions
  i.
\newblock In {\em The Collected Works of Eugene Paul Wigner: Part A: The
  Scientific Papers}, pages 524--540. Springer, 1993.

\bibitem{yang2022developing}
Weixin Yang, Terry Lyons, Hao Ni, Cordelia Schmid, and Lianwen Jin.
\newblock Developing the path signature methodology and its application to
  landmark-based human action recognition.
\newblock In {\em Stochastic Analysis, Filtering, and Stochastic Optimization:
  A Commemorative Volume to Honor Mark HA Davis's Contributions}, pages
  431--464. Springer, 2022.

\end{thebibliography}

\end{document}